\documentclass[10pt,journal,compsoc]{IEEEtran}

\usepackage{algorithm,algpseudocode,graphicx,setspace,amsmath,amssymb,amsthm,url,xcolor,ragged2e}

\newtheorem{defi}{Definition}
\newtheorem{thm}{Theorem}

\def\h{{\bf h}}

\def\g{{\bf g}}
\def\dDelta{{\bf \Delta}}
\def\R{{\bf R}}

\begin{document}



\title{Performance Optimization of \\Network Coding Based Communication and Reliable Storage in Internet of Things}

\author{Jian~Li, Yun~Liu, Zhenjiang~Zhang, Jian~Ren, and Nan~Zhao
\IEEEcompsocitemizethanks{
\IEEEcompsocthanksitem  Jian~Li, Yun~Liu, Zhenjiang~Zhang are with the School of Electronic and Information Engineering, Beijing Jiaotong University, Beijing 100044, China. Email: \{lijian,liuyun,zhjzhang1\}@bjtu.edu.cn
\IEEEcompsocthanksitem Jian~Ren is with the Department of Electrical \& Computer Engineering, Michigan State University, East Lansing, MI 48824, United States. Email: renjian@egr.msu.edu
\IEEEcompsocthanksitem Nan~Zhao is with the School of Telecommunications Engineering, Xidian University, Xi'an, Shanxi 710071, China. Email: zhaonan@xidian.edu.cn
}}


\IEEEtitleabstractindextext{%
\begin{abstract}
Internet or things (IoT) is changing our daily life rapidly. Although new technologies are emerging everyday and expanding their influence in this rapidly growing area, many classic theories can still find their places. In this paper, we study the important applications of the classic network coding theory in two important components of Internet of things, including the IoT core network, where data is sensed and transmitted, and the distributed cloud storage, where the data generated by the IoT core network is stored. First we propose an adaptive network coding (ANC) scheme in the IoT core network to improve the transmission efficiency. We demonstrate the efficacy of the scheme and the performance advantage over existing schemes through simulations. 
Next we introduce the optimal storage allocation problem in the network coding based distributed cloud storage, which aims at searching for the most reliable allocation that distributes the $n$ data components into $N$ data centers, given the failure probability $p$ of each data center. Then we propose a polynomial-time optimal storage allocation (OSA) scheme to solve the problem. Both the theoretical analysis and the simulation results show that the storage reliability could be greatly improved by the OSA scheme.
\end{abstract}

}

\maketitle

\IEEEdisplaynontitleabstractindextext

\IEEEpeerreviewmaketitle

\IEEEraisesectionheading{\section{Introduction}\label{sec:introduction}}


\IEEEPARstart{I}{nternet} of things (IoT)~\cite{fangyg} is an integral part in today's development of smart city. People could remotely access and interact with a wide range of devices integrated with sensors, from home appliances, wearable electronics to environmental monitors. Many new applications of smart city rely on the deployment of Internet of things, such as home automation, remote healthcare, intelligent transportation, smart grid and so on. The high level view of Internet of things is shown in Fig.~\ref{fig:nc-smartcity},  which includes IoT core network for data sensing and transmission, distributed cloud storage~\cite{info14} for storing the data generated by the core network, cloud computing~\cite{cloudcomp} for processing the data. Upon these components are various applications such as e-transportation, e-heath, smart home and so on. Communication networks such as 4G and 5G networks~\cite{fiveg} interconnect these major components.

\begin{figure}[h]
\centering
\includegraphics[width=.95\columnwidth]{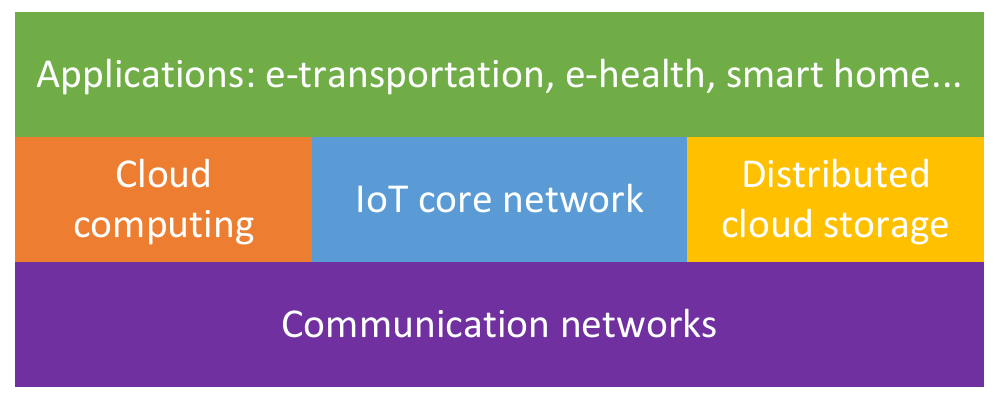}
\caption{High level view of Internet of things}
\label{fig:nc-smartcity}
\end{figure}

The IoT core network is responsible for generating the data for Internet of things. Smart devices sense various data and send out the data through the networks constituted by these devices~\cite{zigbee,sixlowpan}. Since the communication of the IoT core network is largely through wireless, the packet loss may be high due to fadings and interferences. Thanks to the emerging of software defined wireless networking~\cite{sdwn}, we can apply sophisticated algorithms to improve the communication quality of the IoT core network.

During the operation of Internet of things, data collected from a vast number of sensors in the IoT core network could explode. The distributed cloud storage is the best candidate to safely and reliably store these data.  The distributed data storage architecture model distributes the database to multiple servers in many locations across the participating network in the storage cloud. Each location is directly and independently plugged into the Internet. If something unexpected happens to the data in one location, generally only a small amount of backed up data is impacted and the original data could be recovered using the data stored in rest of the locations. In addition, since the data stored is distributed to many locations, users can access the data simultaneously from several locations to efficiently make use of the bandwidth. 

Network coding provides a trade-off between communication capacity and computational complexity in network environment by enabling the intermediate relay nodes to encode the incoming packets before forwarding them. The throughput and robustness of the network can be improved through network coding. In this paper, we study two applications of the classic network coding theory in Internet of things. The main contribution of this paper is:
\begin{itemize}
\item We propose an adaptive network coding scheme (ANC scheme) for the IoT core network and demonstrate that the scheme can improve the transmission efficiency and the performance is better than existing schemes.
\item For the distributed cloud storage utilizing network coding that stores the data generated by the IoT core network, we introduce the optimal storage allocation problem and propose an optimal storage allocation (OSA) scheme. Simulation results show that the storage reliability can be greatly improved.
\end{itemize}

The paper is organized as follows:  in Section~\ref{sec:relatedwork} we briefly review the IoT core network and the distributed cloud storage. The concept of network coding and its advantages in communication and storage are also introduced in Section~\ref{sec:relatedwork}. Next we propose and analyze our adaptive network coding scheme for the IoT core network in Section~\ref{sec:anc}. After that we study the optimal storage allocation problem in the distributed cloud storage utilizing network coding in Section~\ref{sec:osa}. At last is the Conclusion.


\section{Preliminaries and Related Work}\label{sec:relatedwork}
\subsection{Internet of Things}
The objective of Internet of things is to equip everything related to human beings with smart chips integrating sensors, actuators and transceivers. Anything equipped with the smart chip can be called a smart device. Smart devices within a certain range can communicate with each other and form networks with different purposes, such as smart home appliance networks, smart surveillance camera networks, etc. These networks can be further connected to the Internet through proper interconnecting. The benefit of the deployment of Internet of things is obvious. Below are some of the application cases: home owners could remotely monitor and control their home appliance; city residents could check current air pollution levels of any streets; transportation department could make quick actions according to real-time traffic monitoring. Internet of things is the fundamental building block of smart city. 
In this paper, we will mainly focus on the IoT core network and the distributed cloud storage which stores the data generated by the IoT core network as shown in Fig.~\ref{fig:nc-smartcity}. 

\subsubsection{IoT Core Network}
IoT core network consists of the smart devices mentioned above and the networks among these devices. Several challenges of the IoT core network exists.
\begin{itemize}
\item The first challenge is that it lacks of a unified infrastructure and protocol stack. Different academical research groups, industrial R$\&$D teams and standard organizations have proposed different solutions to integrate the smart devices. As a result, smart devices from different vendors cannot communicate with each other, or can communicate only after some complicated bridging work. This has created an obstacle for the realization of Internet of things. As an example, Zigbee~\cite{zigbee} and 6LoWPAN~\cite{sixlowpan} are two popular protocol stacks based on 802.15.4 physical layer and both have been widely adopted in the IoT core network, but they are not compatible with each other. 
\item The second challenge is that the monitor and control of the network lacks of flexibility. It is difficult for the network operator to update network management policies.
\item The third challenge is that the functionality of the network cannot be changed without reprogramming the smart devices when the application environment changes.
\end{itemize}

To overcome these shortcomings, software defined wireless networking (SDWN)~\cite{sdwn} was proposed based on the paradigm of software defined networking (SDN)~\cite{sdn}. The major difference between SDN and SDWN is that in the context of SDWN, the network elements in the data plane are smart devices instead of switches. The smart devices act as both end users and switches. The data flow is separated from the control flow. We can easily change the network behaviors through exchange of the control flow among smart devices. Take the wireless sensor networks in the IoT core network as an example, the layers above the 802.15.4 MAC layer will be defined through software and can be changed instantly to meet the new requirements. To update the network management policies or to change network functionality will be as easy as installing a new software application.

Thanks to the advantages of SDWN, it is much easier to implement algorithms which can improve network performance into IoT core network. In~\cite{sbmedard} the authors propose to combine network coding and software defined networking, where the code rate of the network coding is fixed. Although this approach can improve the communication throughput, the strategy is not flexible to cope with the changing channel qualities in wireless environments. In this paper, we will show that the transmission efficiency of the IoT core network can be greatly improved through our adaptive network coding scheme, where the code rate of the network coding can be dynamically adjusted in a centralized manner with the global view of the whole network.

\subsubsection{Distributed Cloud Storage}
According to the estimation of UNECE, in 2015 the amount of all global data is about 7 zettabytes ($7 \times 10^{21}$). The volume of data will be boosted dramatically with the developing of Internet of things, where there will be hundreds of thousands of sensors deployed to create more and more data, including the air pollution levels of every street of the smart city, health condition of every elder in the smart city, videos captured by surveillance cameras at every corner of the smart city, etc. To the year 2020, the amount of data would grow to 40 zettabytes. How to properly store the data has become a major challenge in Internet of things. The data center will be the backbone for Internet of things and must fulfill the following requirements:

\begin{itemize}
\item The file stored in the data center must be reliable. If the file become unavailable because of hardware failures, the data center should be able to recover the file as soon as possible.
\item The data storage efficiency should be high. This means that the data center could store more data with the same storage devices. This requirement is essential because of the astonishing high volume of the data produced by the IoT core network.
\item The data center should be able to provide confidentiality for the files stored. Some of the files such as medical profiles are related to personal privacy and should not be accessed by unauthorized individuals.
\item The data storage capacity must be scalable. If the data to be stored has exceed the system limit, the data center should be easily updatable to meet the new requirement.
\end{itemize}

Traditional centralized data center is not suitable for the context of Internet of things. If something unexpected happens such as power outage or military actions, the precious data stored in the data center could be lost and unrecoverable. To ensure a high reliability of the data storage, a typical solution is to store the data across multiple servers in the distributed cloud storage. The main idea is that instead of storing the entire data in one server, we can split the data into $n$ data components and store the components separately. The original data can be recovered only when the required (threshold) number of components, say $k$, are collected. The storage efficiency is much higher than simply replicating the data over multiple servers. The original data is information theoretically secure for anyone who can access either an individual component or multiple components when the number of components combined is less than the threshold $k$. In this case, when the individual components are stored distributively across multiple cloud storage servers, each cloud storage server only needs to assure data integrity and data availability. 
The requirement for costly data encryption and secure key management might be eased.
The distributed cloud storage can also increase data availability while reducing network congestion, thus leading to increased resiliency. A popular approach is to employ an $(n, k)$ maximum distance separable (MDS) code, such as the Reed-Solomon (RS) code in the Total Recall system~\cite{Total}.

From the analysis above we can see that the distributed cloud storage is essential in the development of Internet of things since it meets all the requirements of the data center in Internet of things. In later sections, we will show that by applying network coding in the distributed cloud storage we can further improve the performance of data storage.

\subsection{Network Coding}
In this section, we will briefly introduce the concept of network coding and its advantages in improving the communication throughput and the distributed cloud storage performance. Network coding has shown its benefits in traditional communication/storage networks and can be further applied in Internet of things. Network coding was first introduced in the seminal paper by Ahlswede \emph{et al.} in~\cite{Ahlswede}. By allowing the intermediate relay nodes to encode the incoming packets, the network could achieve the maximum multicast capacity.

A network is equivalent to a directed graph $G=(V,E)$, where $V$ represents the set of vertices corresponding to the network nodes and $E$ represents all the directed edges between vertices corresponding to the communication link. The start vertex $v$ of an edge $e$ is called the tail of $e$ and written as $v=tail(e)$, while the end vertex $u$ of an edge $e$ is called the head of of $e$ and written as $u= head(e)$. For a source node $u$, there is a set of symbols $\mathcal{X}(u) = (x_1,\dots,x_k)$ to be sent. Each of the symbol is from the finite field $GF(2^m)$, where $m$ is a positive integer. For a link $e$ between intermediate nodes $r_1$ and $r_2$, written as $e=(r_1,r_2)$, the symbol $y_e$ transmitted on it is the function of all the $y_{e'}$ such that $head(e')=r_1$. And $y_e$ can be written as:
\begin{equation}
y_e = \sum_{e':head(e')=r_1}\beta_{e',e} \cdot y_{e'},
\end{equation}
in which the encoding coefficients $\beta_{e',e}\in GF(2^m)$. For a sink node $v$, there is a set of incoming symbols $y_{e'}~(e':tail(e')=v)$ to be decoded. 

\subsubsection{Network Coding in Communication}
The main idea of network coding can be illustrated through Fig.~\ref{fig:nc-example}. Assume the capacity of all the edges is $C$, the capacity of this network is $2C$ according to the max-flow min-cut theorem. Only by encoding the incoming packet symbols $x_1,x_2$ at node R3, this network can achieve the maximum capacity.

In~\cite{Koetter,Ho} the authors have shown that linear codes with random selected coefficients are sufficient to achieve the multicast capacity by coding on a large enough field. Sink nodes that have received more linear independent encoded symbols than the original symbol generated by the source nodes can easily decode the original symbols by solving a set of linear equations. Moreover, it has been demonstrated that network coding can improve the communication throughput. As an example, the authors in~\cite{Gkantsidis} have applied the principles of random network coding to the context of peer-to-peer (P2P) content distribution, and have shown that file downloading times can be reduced. Thus in this paper we propose to apply adaptive random linear network coding in the IoT core network to improve the network transmission efficiency.

\begin{figure}
\centering
\includegraphics[width=1.0\columnwidth]{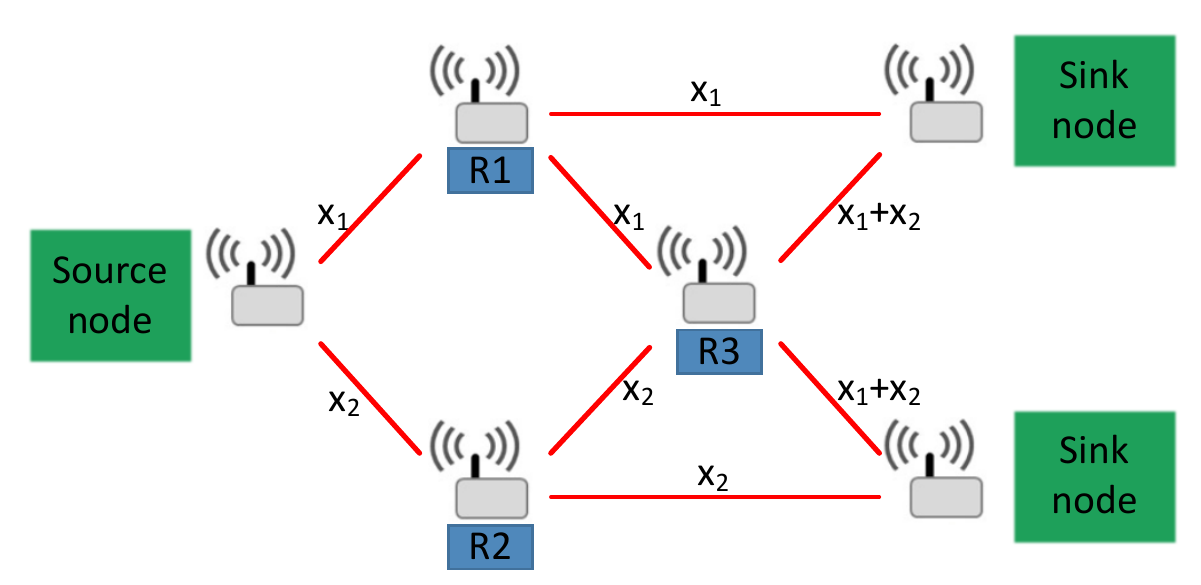}
\caption{A simple example of network coding}
\label{fig:nc-example}
\end{figure}

\subsubsection{Distributed Cloud Storage Utilizing Network Coding}
When a storage node in the distributed cloud storage network that employing $(n, k)$ RS code (such as Total Recall~\cite{Total}) fails, the replacement node connects to $k$ nodes and downloads the data of the same amount as the whole file first to decode the original file. Then the replacement node encodes the original file using the same $(n, k)$ code to recover the encoded part of the file stored in the failed node. This approach is a waste of bandwidth because the whole file has to be downloaded to recover a fraction of it.

To overcome this drawback, Dimakis \emph{et al.}~\cite{Dimakis} introduced the conception of $(n,k,d,\alpha,\beta,B)$ regenerating code based on the network coding. In the context of regenerating code, the contents stored in a failed node can be regenerated by the replacement node through downloading $\beta$ help symbols from each of $d$ helper nodes. This regeneration is identical to the encoding process of the intermediate nodes in network coding. The bandwidth consumption for the failed node regeneration could be far less than the whole file. A data collector (DC) can reconstruct the original file stored in the network by downloading $\alpha$ symbols from each of the $k$ storage nodes. In~\cite{Dimakis}, the following theoretical bound was derived based on network coding theory:
\begin{equation}
\label{eq:min_cut}
B \leq \sum_{i=0}^{k-1}\min \{ \alpha, (d-i)\beta \}.
\end{equation}
From equation~(\ref{eq:min_cut}), a tradeoff between the regeneration bandwidth $\gamma = d\beta$ and the storage requirement $\alpha$ was derived. $\gamma$ and $\alpha$ cannot be decreased at the same time. There are two special cases: minimum storage regeneration (MSR) point in which the storage parameter $\alpha$ is minimized:
\begin{equation}
\label{eq:MSR_tradeoff}
(\alpha_{MSR},\gamma_{MSR})= \left(\frac Bk, \frac{Bd}{k(d-k+1)}\right),
\end{equation}
and minimum bandwidth regeneration (MBR) point in which the bandwidth $\gamma$ is minimized:
\begin{equation}
\label{eq:MBR_tradeoff}
(\alpha_{MBR},\gamma_{MBR})= \left(\frac{2Bd}{2kd-k^2 + k},\frac{2Bd}{2kd-k^2 + k} \right).
\end{equation}

\begin{figure}
\centering
\includegraphics[width=1.0\columnwidth]{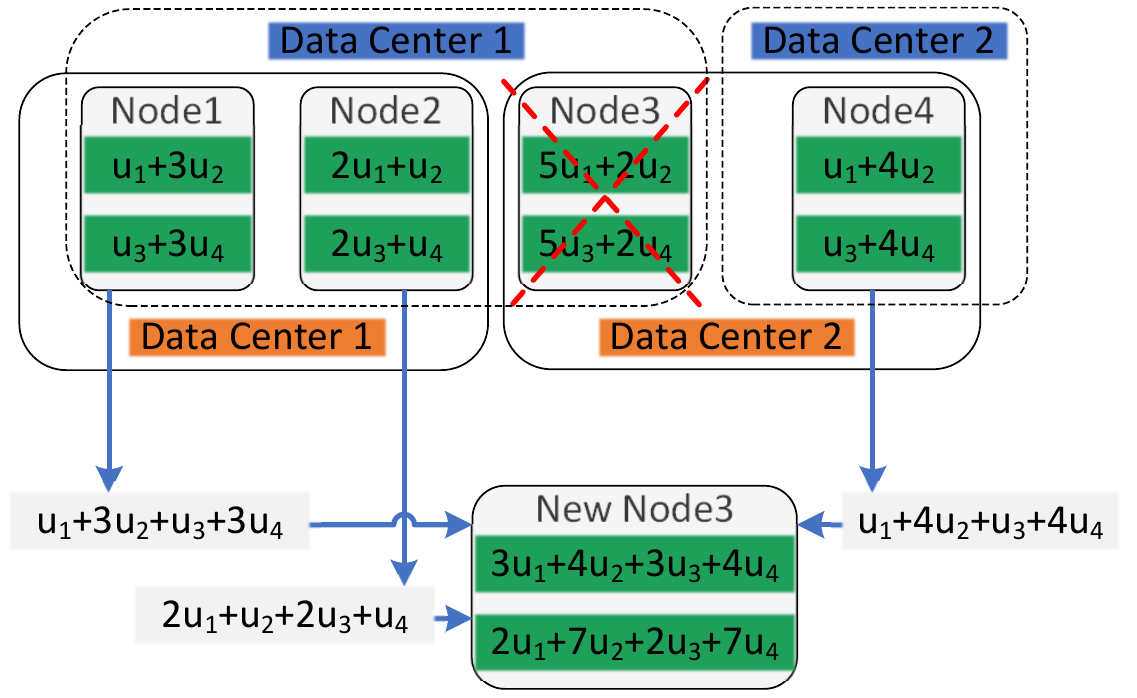}
\caption{network coding based distributed cloud storage}
\label{fig:nc-cloud}
\end{figure}

Fig.~\ref{fig:nc-cloud} is an illustrative example of regenerating code with parameters $n=4,k=2,d=3,\alpha=2,\beta=1,B=4$. $4$ symbols $u_1,u_2,u_3,u_4$ are stored in $4$ storage nodes, and can be retrieved from any $2$ of the storage nodes. A failed node can be regenerated by downloading $1$ symbol each from the $3$ remaining nodes. Here we suppose node 3 fails. For the storage systems simply employing RS code, $4$ symbols have to be downloaded first to decode the original symbols. Then we have to encode the $4$ decoded symbols again to regenerate the symbols in the failed node 3. So the bandwidth needed for repairing the failed node 3 is $4$. For the regenerating code solution in Fig.~\ref{fig:nc-cloud}, by linearly combing the $3$ downloaded symbols $u_1+3u_2+u_3+3u_4$, $2u_1+u_2+2u_3+u_4$ and $u_1+4u_2+u_3+4u_4$ into $2$ symbols $3u_1+4u_2+3u_3+4u_4$ and $2u_1+7u_2+2u_3+7u_4$, we can regenerate a new node 3 that has the same function as the failed node 3. In the repairing process, only $3$ symbols need to be downloaded. Thus the repair bandwidth is saved by $25\%$. In the later section, we will introduce a storage allocation problem for regenerating code in the distributed cloud storage and propose an optimal storage allocation scheme that can achieve the highest possible reliability.


\section{Adaptive Network Coding in the IoT core network}\label{sec:anc}
In this section, we will show our adaptive network coding (ANC) scheme in the IoT core network. The size of the data to be transmitted in the IoT core network may be larger than the size limit of a single packet, such as new firmwares to update the smart devices on-air. So the data needs to be divided into data fragments first then transmitted in multiple packets with one data fragment per packet. A node has to correctly receive enough linearly independent packets to reassemble the original data. 

\subsection{Limitations of existing works}
Since the communication between smart devices are through wireless channel and there may be various fadings and interferences in the channel, some of the nodes may experience packets loss in the communication. When network coding is not utilized, retransmission is a common method to mitigate the packets loss. In some cases, certain packets may get lost most of the time so these packets have to be retransmitted many times until they are correctly received. Thus the overall transmission efficiency will be low. {\it{Here the transmission efficiency is defined as the ratio between the minimum number of the packets needed to reassemble the original data and the number of total packets transmitted from the source node and the intermediate nodes}}. When network coding is utilized, a node can retrieve the original data as long as the node can correctly receive enough number of packets. The entire transmission of the data will not be affected by lacking of certain particular packets. So the overall transmission efficiency will be higher. 

However, there are still limitations for simply applying the network coding in the IoT core network, 
where the number of encoded packets to be generated and sent in the intermediate nodes is predetermined~\cite{sbmedard}. If too few packets are generated, the sink node may not even be able to collect enough packets to decode the original data. If too many packets are generated, the transmission efficiency will be low. Moreover, the fact that the quality of wireless channel is changing over time makes the situation even worse. As an example, when the channel quality becomes better and the packet loss rate goes lower, some of the encoded packets will be useless and the transmission efficiency could be higher. The encoding strategy should be able to dynamically adjusted according to the transmission conditions.

\subsection{ANC Scheme for the IoT core network}

\begin{figure}[h!]
\centering
\includegraphics[width=.9\columnwidth]{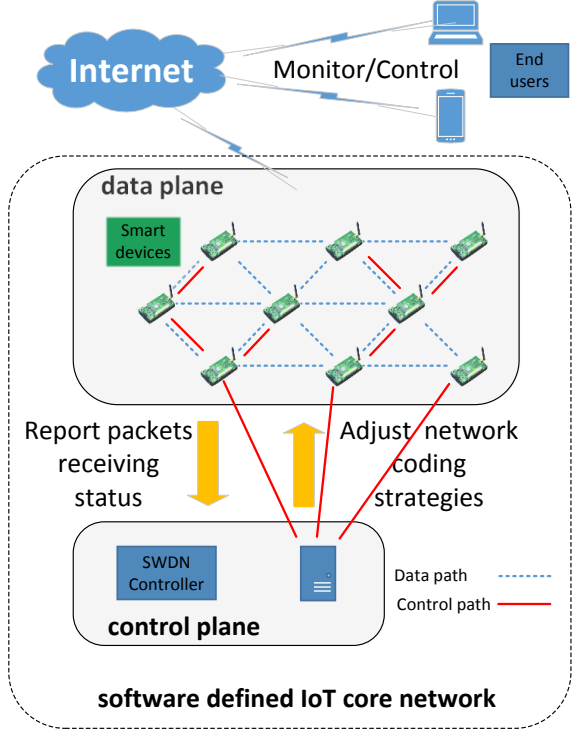}
\caption{ANC Scheme for the IoT core network}
\label{fig:nc-sdwn}
\end{figure}

To overcome the limitations mentioned above, we propose an adaptive network coding (ANC) scheme to further improve the transmission efficiency of the IoT core network with SDWN, as illustrated in Fig.~\ref{fig:nc-sdwn}. In this example, end users could communicate with the IoT core network with SDWN through Internet for monitoring/control purposes. The data transmission inside the IoT core network could benefit from our ANC scheme. In the figure, the IoT core network formed by the smart devices could be the smart appliance network at home, the surveillance camera network on streets or the emission detecting network in factories, etc. Here we only include the data plane and control plane of the SWDN to show the main idea of the ANC scheme. 

For the data plane, the source node will send out linear combinations of the original packets. Each intermediate node will perform random linear network coding. The incoming packets will be linearly combined using random coefficients then sent out to succeeding nodes. \emph{The code rate $r$ of the network coding is defined as the ratio of the number of encoded packets to the number of incoming packets.} And the code rates of the network coding will be automatically adjusted by the SWDN controller mentioned below. The sink nodes will decode the original packets after receiving enough number of linearly independent packets. 

Meanwhile, for the control plane, the smart devices will report packets receiving statistics to the SDWN controller periodically through the control path. Based on the information reported, the SDWN controller will dynamically adjust the network coding strategies to eliminate unnecessary transmissions. If the packet loss becomes higher around some node, more encoded packets will be generated in the corresponding intermediate nodes. If the packet loss becomes lower, the number of encoded packets will be decreased. Since the SDWN controller has the global information of the network, this centralized control will be more effective.

\subsubsection{Source Node Algorithm of the ANC scheme}
In the source node, the data to be transmitted will be fragmented into data packets with equal length. Every $n$ data packets will form a coding group, in which random linear network coding will be performed. For the purpose of clarity, in the paper we assume that there is only one coding group. For each packet $\h_i$ in the coding group, there will be an encoding vector ${\dDelta_i} = [\delta_{i,1}, \delta_{i,2}, \dots, \delta_{i,n}]$ ($\delta_{i,j} \in GF(2^m), 1 \leq i,j \leq n$) attached in front to indicate which packets participate in the encoding of $\h_i$. 
$GF(2^m)$ denotes the finite field with $2^m$ elements where $m \in \{8,16,32,64,\dots\}$ is determined by the symbol size. For an uncoded packet $\h_i$, the elements in the encoding vector will be all-zero except $\delta_{i,i}=1$. The packet format is illustrated in Figure~\ref{fig:pacfor}. The source node will perform Algorithm~\ref{alg:source} to send out the encoded packets $\g_i$ ($1 \leq i \leq \lceil rn \rceil$) where $r$ is the code rate determined by the SWDN controller and $\lceil rn \rceil$ is the ceiling operation to get the smallest integer that is larger than or equal to $rn$.

\begin{figure}
\centering
\includegraphics[width=.8\columnwidth]{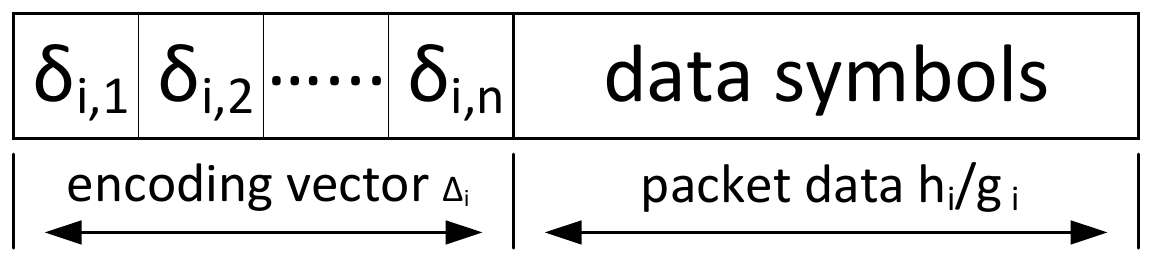}
\caption{ANC packet format}
\label{fig:pacfor}
\end{figure}

\begin{algorithm}                                
\begin{spacing}{1.0}   
\begin{algorithmic}
    \State \Comment the network coding code rate $r$ is determined/updated by the SWDN controller 
    \For{$i=1 \to \lceil rn \rceil$}   
    \If {$i \leq n$}
    \Repeat{$\:\:\:$ generate a random encoding vector $\dDelta_i=$ \\ \hskip65pt $[\delta_{i,1}, \delta_{i,2}, \dots, \delta_{i,n}]$ }
    \Until{$\:\:\:$ $\dDelta_i$ is linearly independent from all the $\dDelta_j,$ \\ \hskip60pt $1 \leq j < i$ (except for $i=1$)}              
    \Else
    \State generate a random encoding vector $\dDelta_i =[\delta_{i,1}, \delta_{i,2}, $ \\ \hskip30pt $\dots, \delta_{i,n}]$   
    \EndIf
    \State \Comment  multiply each symbol of packet data $\h_j$ by $\delta_{i,j}$ $\:\:\:\:\:\:\:\:\:\:\:\:$
    \State $\g_i \Leftarrow \sum_{j=1}^{n}\delta_{i,j}\h_j\:\:\:$ 
    \State send out $\left [\dDelta_i || \g_i \right ]$, where``$||$'' is the concatenation oper-\\ \hskip15pt ation
    \EndFor
    \State report the number of sent out packets to the SWDN controller
\end{algorithmic}
\end{spacing}
\caption{ANC scheme - source node}          
\label{alg:source}                  
\end{algorithm}

In algorithm~\ref{alg:source}, the source node generates $\lceil rn \rceil$ $n$-dimensional encoding vectors $\dDelta_i$ (first $n$ encoding vectors are linearly independent) and uses the vector elements $\delta_{i,1}, \delta_{i,2}, \dots, \delta_{i,n}$ as coefficients to generate and send out encoded packets from the uncoded packets $\h_1,\dots,\h_n$.

\subsubsection{Intermediate Node Algorithm of the ANC scheme}
For each coding group, the intermediate node will open a receiving buffer to store the incoming fresh packets from the nodes designated by the SDWN controller for encoding. The intermediate node will also record all the encoding vectors received in the incoming packets. \emph{A packet is called a fresh packet if its encoding vector is linearly independent from all of the previously received packets'.} In order to get a trade-off between the packet diversity and communication delay, the intermediate node will encode the incoming fresh packets received during a preset interval $\tau$ which is measured by a timer then clear the receiving buffer and wait for the next incoming fresh packet to restart the timer and the buffering. At the end of each time interval, the encoding of the fresh packets in the receiving buffer is performed. For better illustration, we can split each of the $n_{\tau}$ fresh packets in the receiving buffer into the encoding vector $\dDelta_i$ and data $\g_i$  ($1 \leq i \leq n_{\tau}$). $n_{\tau}$ is the number of fresh packets in the receiving buffer. The intermediate node will send out  $\lceil rn_{\tau} \rceil$ encoded packets using Algorithm~\ref{alg:int}, where $r$ is the code rate determined by the SWDN controller. At the same time, the intermediate node will report the receiving and the sending of the packets to the SWDN controller.

\begin{algorithm}                                
\begin{spacing}{1.0}   
\begin{algorithmic}
    \State \Comment the network coding code rate $r$ is determined/updated by the SWDN controller
    \For{$i=1 \to \lceil rn_{\tau} \rceil$}   
    \If {$i \leq n_{\tau}$}
    \Repeat{$\:\:\:$ generate a random vector $\R_i=[r_{i,1}, r_{i,2},$ \\ \hskip65pt $ \dots, r_{i,n_{\tau}}]$}
    \Until{$\:\:\:$ $\R_i$ is linearly independent from all the $\R_j,$ \\ \hskip62pt  $1 \leq j < i$ (except for $i=1$)}              
    \Else
    \State generate a random vector $\R_i=[r_{i,1}, r_{i,2}, \dots, r_{i,n_{\tau}}]$   
    \EndIf
    \State \Comment multiply each symbol of $\dDelta_j$ by $r_{i,j}$ $\:\:\:\:\:\:\:\:\:\:\:\:\:\:\:\:\:\:\:\:\:\:\:\:\:\:\:\:\:\:\:\:\:$
    \State $\dDelta'_i \Leftarrow \sum_{j=1}^{n_{\tau}}r_{i,j} \dDelta_j \:\:\:$ 
    \vskip5pt
    \State \Comment multiply each symbol of $\g_j$ by $r_{i,j}$ $\:\:\:\:\:\:\:\:\:\:\:\:\:\:\:\:\:\:\:\:\:\:\:\:\:\:\:\:\:\:\:\:\:$
    \State $\g'_i \Leftarrow \sum_{j=1}^{n_{\tau}}r_{i,j} \g_j \:\:\:$ 
    \vskip5pt
    \State send out $\left [\dDelta'_i || \g'_i\right ]$, where``$||$'' is the concatenation oper- \\ \hskip15pt ation
    \EndFor
    \State report the number of received packets from each of the other nodes/the number of sent out packets to the SWDN controller    
\end{algorithmic}
\end{spacing}
\caption{ANC scheme - intermediate node}          
\label{alg:int}                  
\end{algorithm}

In algorithm~\ref{alg:int}, the intermediate node generates $\lceil rn_{\tau} \rceil$ $n_{\tau}$-dimensional vectors $\R_i$ (first $n_{\tau}$ vectors are linearly independent) and uses the vector elements $r_{i,1}, r_{i,2}, \dots, r_{i,n_{\tau}}$ as coefficients to generate and send out recoded packets from the received packets $\g_1,\dots,\g_{n_{\tau}}$. The corresponding encoding vectors are processed the same way.

\subsubsection{Sink Node Algorithm of the ANC scheme}
Once the sink node receives $n$ linearly independent packets, it can solve the following equation to decode the original packets data $\h_1, \h_2, \dots, \h_n$:
\begin{equation}
\label{eq:encoding_msr_h}
\begin{bmatrix}
\dDelta_1 \\
\dDelta_2\\
\vdots \\
\dDelta_n
\end{bmatrix}
\begin{bmatrix}
\h_1 \\
\h_2\\
\vdots \\
\h_n
\end{bmatrix}
= 
\begin{bmatrix}
\g_1 \\
\g_2\\
\vdots \\
\g_n
\end{bmatrix}.
\end{equation}
Then $\h_1, \h_2, \dots, \h_n$ can be concatenated to restore the original data. The sink node will also periodically report the packets receiving status to the SWDN controller.

\subsubsection{SDWN Controller Algorithm of the ANC scheme}
Since the SDWN controller receives the packets sending/receiving status from each of the nodes in the IoT core network periodically , it can adjust the code rate of the network coding for each of the nodes accordingly. Suppose $\eta_i$ is the number of packets sent by node $i$, $\mathbb{N}_i$ is the set of succeeding nodes receiving the packets from node $i$, and $\eta_j^{(i)}$ is the number of packets received by node $j \in \mathbb{N}_i$. The code rate $r_i$ of node $i$ can be determined by
\begin{equation}
r_i = \frac{\eta_i}{\max_{j \in \mathbb{N}_i}(\eta_j^{(i)})},
\end{equation}
where $\max()$ is the operation to select the maximum element. Besides the code rate, since the SDWN controller has all of the topology information, for each node $i$, it can specify the succeeding relay nodes to receive the packets sent from node $i$. 

\begin{algorithm}                                
\begin{spacing}{1.0}   
\begin{algorithmic}
    \For {each of the source node or intermediate node $i$}
    \State calculate $r_i$ according to the packet sending/receiving \\ \hskip15pt status of node $i$ and $\mathbb{N}_i$
    \State send $r_i$ to update the code rate of node $i$ through the \\ \hskip15pt control path
    \EndFor   
\end{algorithmic}
\end{spacing}
\caption{ANC scheme - SDWN controller}          
\label{alg:sdwn}                  
\end{algorithm}

\subsection{Performance Evaluation of the ANC scheme}

\begin{figure}
\centering
\includegraphics[width=1.05\columnwidth]{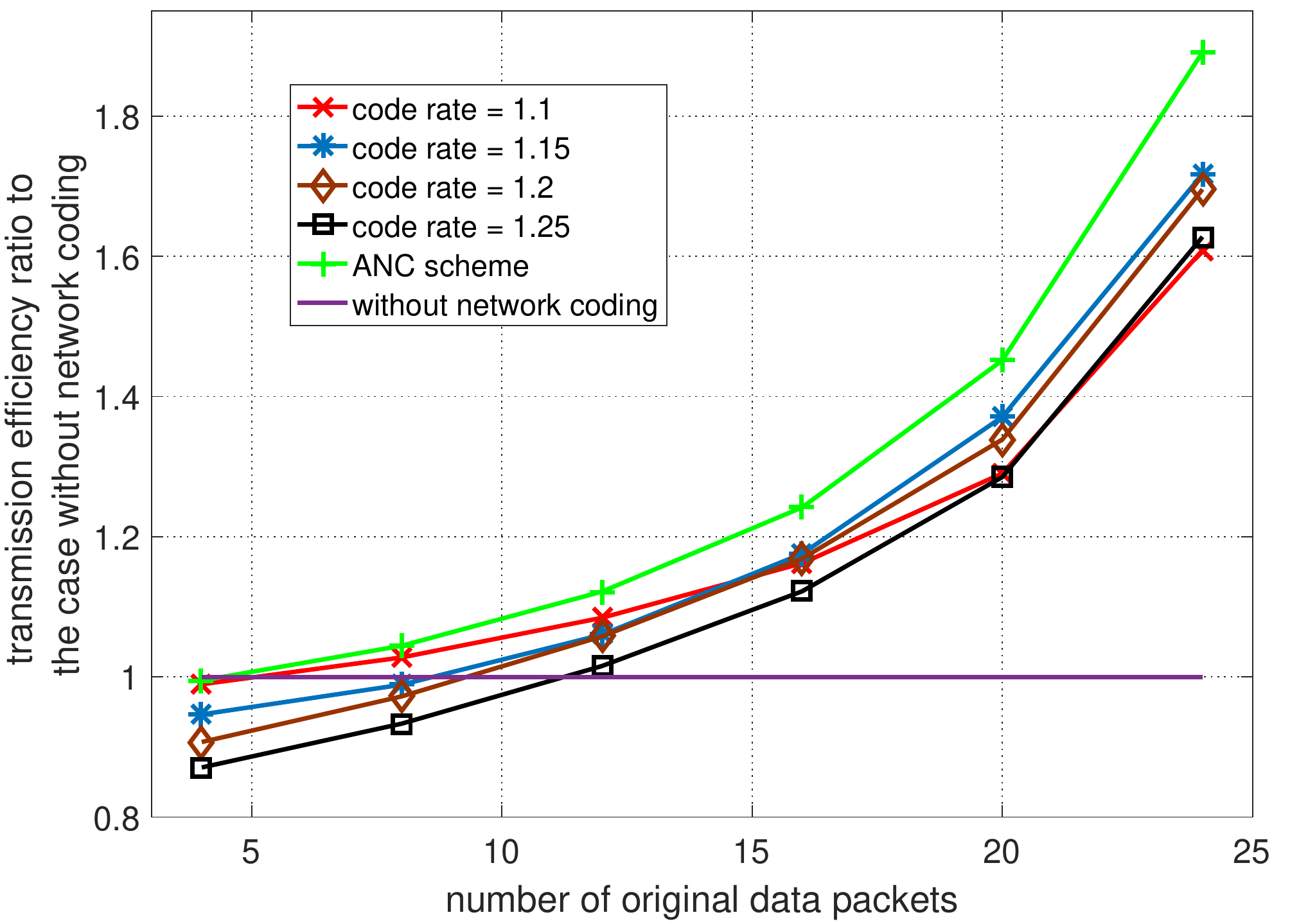}
\caption{Performance of the ANC scheme}
\label{fig:nc-simulation}
\end{figure}

In Fig.~\ref{fig:nc-simulation} are the simulation results of the ANC scheme. The simulation is carried out in the NS-2 platform. In the simulation, the leftmost node in Fig.~\ref{fig:nc-sdwn} tries to send data to the rightmost two nodes. The qualities of the channels between the intermediate nodes are chosen randomly. We calculate the transmission efficiencies under different numbers of total original data packets (can be viewed as one coding group for network coding). For performance comparison, we also simulate the cases for no network coding (pure retransmission) and network coding with predetermined code rates as in~\cite{sbmedard}. 
To make the comparison more clear, we normalize the transmission efficiencies for the cases with fixed network coding and the ANC scheme by the transmission efficiency for the case without network coding. From the simulation results, we can see that the transmission efficiency for the case with network coding becomes higher than the case without network coding with the increasing of the number of original data packets. And the ANC performs best among all the cases. It can also be seen that the performance gain of the ANC scheme will increase when the number of total original data packets becomes larger. 

\section{Optimal Storage Allocation in the Distributed Cloud Storage Utilizing Network Coding}\label{sec:osa}
In this section we first introduce a storage allocation problem for the distributed cloud storage utilizing network coding. Then we propose the optimal storage allocation (OSA) scheme. We also show the performance of the optimal storage allocation scheme.

\subsection{Storage Allocation Problem}
In the storage allocation problem $\mathbb{S}$, the data is encoded with an $(n,k,d,\alpha,\beta,B)$ regenerating code, so there will be $n$ encoded parts. There are $N$ ($N<n$) data centers in total to store these parts, each with a failure probability of $p$. If a data center fails, all the data stored in the data center will be lost. If the total number of encoded parts in the remaining data centers is less than $k$, the original data cannot be recovered any more. Since there are more encoded parts than the data centers, there will be different allocation strategies of the encoded parts with different storage reliabilities. For the problem $\mathbb{S}$, we try to find out the allocation strategy with the lowest failure probability among all the possible allocation strategies.

\begin{defi}
A set $S$ with $N$ elements $n_1,n_2,\dots,n_N$ ($n_i>0, 1 \leq i \leq N$) is a valid allocation 
if $\sum_{i=1}^{N}n_i=n$.
\end{defi}

\begin{defi}
For an allocation strategy $S$, the failure probability $P$ is defined as the probability that the original data cannot be recovered given the failure probabilities of individual data centers.
\end{defi}

The problem $\mathbb{S}$ can be formulated as: 

\begin{align}
\label{eq:optP}
\begin{split}
\mbox{find}~~~ & \mbox{the allocation $S$ among all the valid allocations},\\ 
\mbox{such that}~~~  &  \sum_{\forall  S_j \subseteq S }P\left(\sum_{n_i\in S_j}n_i \ge n-k\right) \mbox{is minimal}.
\end{split}
\end{align}
As an example, for the regenerating code in Fig.~\ref{fig:nc-cloud}, $n=4$ encoded parts are stored in $N=2$ data centers.   Suppose the failure probability of each data center is $p=0.01$. Two storage allocation strategies are shown in the figure. For the first allocation strategy $S=\{3,1\}$ (blue data centers with dash lines), $3$ encoded parts are stored in data center 1 and $1$ encode part is stored in data center 2. It is easy to calculate the failure probability of this allocation strategy is $0.01$. For the second allocation strategy $S=\{2,2\}$ (orange data centers with solid lines), $2$ encoded parts are stored in each of the two data centers. The failure probability of this allocation strategy is $0.0001$, which is much lower than that of the first strategy.

\subsection{Optimal Storage Allocation Scheme}
In this section, we will show our optimal storage allocation (OSA) scheme to solve the storage allocation problem. The OSA scheme includes two stages: the first stage is to find out all the possible valid allocations $S$ and the second stage is to calculate the failure probability $P$ for each $S$. Then we can output the allocation with the lowest failure probability through comparison. 

\subsubsection{Stage I: Find out All the Possible Valid allocations $S$}
The naive approach to find out all the possible valid $S$ is to search all the possible combinations of $n_1,n_2,\dots,n_N$ such that $\sum_{i=1}^{N}n_i=n$. However, this approach will take exponential time thus is not practical. In our OSA scheme, we first change this problem into an integer partition problem~\cite{comb}: to allocate $n$ encoded parts into $N$ storage centers is the same as to partition an integer $n$ into $N$ parts. Take $n=7,N=3$ as an example, there are $4$ ways to partition $7$ into $3$ parts: $\{1,1,5\}$, $\{1,2,4\}$, $\{1,3,3\}$ and $\{2,2,3\}$, which also consist all the possible valid allocations. Then we can solve the integer partition problem using dynamic programming based on the following recurrence equation:

\begin{equation}\label{eq:part}
\mathbb{P}(n,N) = \mathbb{P}(n-1,N-1) +  \mathbb{P}(n-N,N),
\end{equation}
where $\mathbb{P}(i,j)$ is the total number of ways of partitioning integer $i$ into $j$ parts. The first part of equation~(\ref{eq:part}) is the subproblem where at least one $1$ exists in the partition and the second part of the equation is the subproblem where no $1$ exists in the partition. Thus the solution to the original problem perfectly incorporates these two subproblems, which make it feasible to solve using dynamic programming. We propose Algorithm~\ref{alg:part} to find out all the possible valid allocations $S$. In the algorithm, we use $S(i,j,k)$ to represent the $k^{th}$ valid allocation out of the $\mathbb{P}(i,j)$ allocations for $i$ encoded parts and $j$ storage centers. $\cup$ is the union operation between two sets. The addition between a set $S$ and a number $x$ is defined as the additions between every element of the set and the number: 
\begin{equation}\label{eq:setplus}
S+x := \{n_i + x | n_i \in S \:\:\mbox{for}\:\: 1 \leq i \leq N\}.
\end{equation}
After the execution of the algorithm, we can get all the possible valid allocations $S(n,N,k)$ $(1 \leq k \leq \mathbb{P}(n,N))$. It is easy to see that the algorithm runs in polynomial time.

\begin{algorithm}                                
\begin{spacing}{1.0}   
\begin{algorithmic}[1]               
    \Require the number of encoded parts $n$ and the number of storage centers $N$
    \Ensure all the valid allocations $S(n,N,l)$, $(1 \leq l \leq \mathbb{P}(n,N))$
    \Function{FindAllAllocations}{$n,N$}
    \For{$i=1 \to n$}      
        \State $\mathbb{P}(i,1) \Leftarrow 1$
        \State $S(i,1,1) \Leftarrow i$
        \For{$j=2 \to N$}        
            \If{$i \geq j$}
                \If{$i-j < j$}
                    \State $\mathbb{P}(i,j) \Leftarrow \mathbb{P}(i-1,j-1)$
                    \State $S(i,j,l) \Leftarrow S(i-1,j-1,l) \cup \{1\}$, for \Statex \hskip72pt $1 \leq l \leq \mathbb{P}(i,j)$
                \Else
                    \State $\mathbb{P}(i,j) \Leftarrow \mathbb{P}(i-1,j-1) + \mathbb{P}(i-j,j)$
                    \State $S(i,j,l) \Leftarrow S(i-1,j-1,l) \cup \{1\}$, for \Statex \hskip72pt  all $\:\:1 \leq l \leq \mathbb{P}(i-1,j-1)$                    
                    \State $S(i,j,\mathbb{P}(i-1,j-1) + l) \Leftarrow S(i-j,j,l) +$ \Statex \hskip72pt $1$, for all $\:\:1 \leq l \leq \mathbb{P}(i-j,j)$
                \EndIf
            \EndIf
        \EndFor
    \EndFor
    \EndFunction
\end{algorithmic}
\end{spacing}
\caption{OSA scheme - stage I}          
\label{alg:part}                  
\end{algorithm}

\begin{thm}
Algorithm~\ref{alg:part} can output all the valid allocations  $S(n,N,l)$ for  $1 \leq l \leq \mathbb{P}(n,N)$, where $S(n,N,l)$ represents the $l^{th}$ valid allocation out of the $\mathbb{P}(n,N)$ allocations for $n$ encoded parts and $N$ storage centers.
\end{thm}

\begin{proof}
Algorithm~\ref{alg:part} calculates $S(i,j,l)$ $(1 \leq l \leq \mathbb{P}(i,j))$ for $1 \leq j \leq N$ from $i=1$ to $i=n$ through a bottom-up manner and we can get $S(n,N,l)$ $(1 \leq l \leq \mathbb{P}(n,N))$ for $i=n,j=N$. For each $i$, line 3 to line 4 first calculate $\mathbb{P}(i,1)=1$ and $S(i,1,1)={i}$, corresponding to the case of allocating $i$ encoded data parts into one data center. Then for each $j=2,\dots,N$, there will be two cases:
\begin{itemize}
\item Line 8 to line 9 correspond to the case with $i-j < j$, where at least one storage node will be allocated only 1 encoded data part. The second part of equation~(\ref{eq:part}) does not exist. So the number of ways of allocating $i$ encoded data parts into $j$ storage nodes will be equal to that of allocating $i-1$ encoded data parts into $j-1$ storage nodes: $\mathbb{P}(i,j) =  \mathbb{P}(i-1,j-1)$. And each of the valid allocations $S(i,j,l)$ will be the union of each already calculated allocations $S(i-1,j-1,l)$ with the set $\{1\}$.
\item Line 11 to line 13 correspond to the case with $i-j \geq j$, where $\mathbb{P}(i,j)$ is the summation of two previously calculated parts as shown in equation~(\ref{eq:part}). The computation of the first part and the corresponding valid allocations is the same as in line 8 to line 9. The second part is the number of ways of allocating $i-j$ encoded data parts into $j$ storage nodes $\mathbb{P}(i-j,j)$, where each of the storage node will be allocated at least $2$ encoded data parts. Thus each of the valid allocations $S(i,j, \mathbb{P}(i-1,j-1) + l)$ will be each of the already calculated allocations $S(i-j,j,l)$ plus $1$ as defined in equation~(\ref{eq:setplus}).
\end{itemize}
\end{proof}

Fig.~\ref{fig:partexample} illustrates the algorithm for $n=7$ encoded data parts and $N=3$ data centers. Each $(i,j)$ pair  represent the calculation of $\mathbb{P}(i,j)$ and $S(i,j,l)$. The pairs without shades are calculated using line 8 to line 9 (the first case) while the pairs in shades are calculated using line 11 to line 13 (the second case). The solid lines correspond to the first part of equation~(\ref{eq:part}) and the dashed lines correspond to the second part. From the figure we can clearly see that $(7,3)$ can be efficiently calculated using the results of $(6,2)$ and $(4,3)$, which have already been calculated the same way as illustrated in Fig.~\ref{fig:partexample}.

\begin{figure}
\centering
\includegraphics[width=0.5\columnwidth]{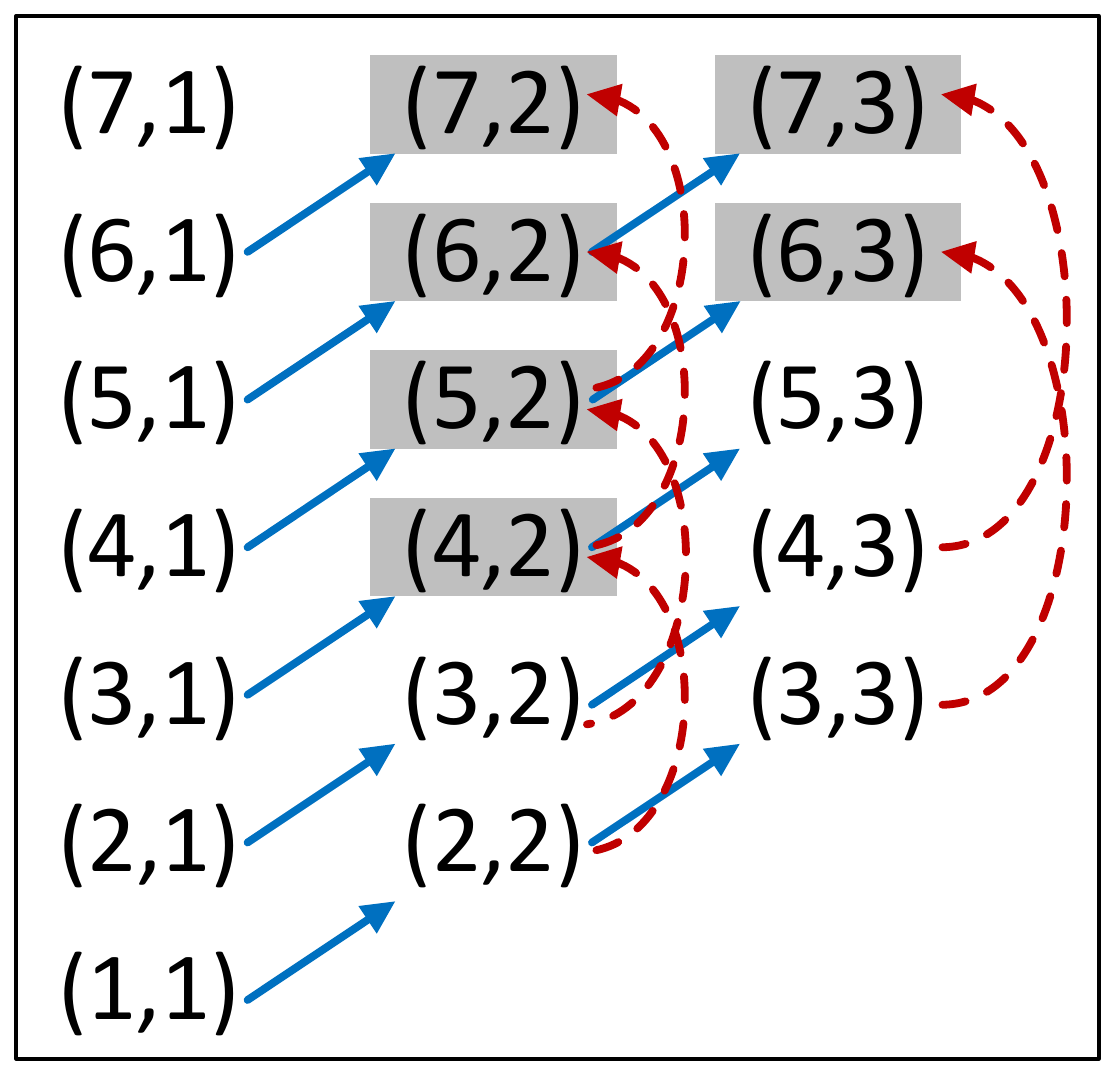}
\caption{The calculation of Algorithm~\ref{alg:part} for $n=7$, $N=3$}
\label{fig:partexample}
\end{figure}

\subsubsection{Stage II: Calculate the Failure Probability $P$ for Each Valid Allocation $S$}
After we get all the possible valid allocations $S$, we can calculate the failure probability $P_S$ for each of them. The goal function of equation~(\ref{eq:optP}) can be further written as:
\begin{align}
\label{eq:optPf}
\begin{split}
P_S  & =   \sum_{\forall  S_j \subseteq S }P\left(\sum_{n_i\in S_j}n_i > n-k\right) \\
 & =\sum_{\forall  S_j \subseteq S,\:\: s.t.\sum_{n_i\in S_j}n_i > n-k}p^{\left |  S_j\right |} (1-p)^{N - \left |  S_j\right |}, 
\end{split}
\end{align}
where $p$ is the failure probability of each storage center, $\left | S_j \right |$ is the number of elements in subset $S_j$. If we try to directly calculate $P_S$ for every subset $S_j \in S$, the order of the number of subsets to be calculated will be approximate to $\sum_{\left | S_j \right |=1}^{N}\begin{pmatrix}N\\\left | S_j \right |\end{pmatrix}\approx 2^{N}$, where $\begin{pmatrix}N\\ \left | S_j \right |\end{pmatrix}$ denotes the number of $\left | S_j \right |$-combinations of the set $S$, thus making it infeasible to calculate in practice. 

In the second stage of the OSA scheme (Algorithm~\ref{alg:calc}), we propose to change the exhaust search problem into a number counting problem. More specifically, for each $i$ ($1 \leq i \leq N$), we count the total number of subsets $ S_j^{(i)}$ such that $ S_j^{(i)}$ denotes the subsets with exactly $i$ elements and the summation of every element in $ S_j^{(i)}$ is larger than $n-k$:

\begin{equation}\label{eq:ps}
P_S =  \sum_{i=1}^{N} \left | \left  \{  S_j^{(i)}\:\:  | \sum_{n_i\in S_j^{(i)}}n_i > n-k \right \} \right | p^{  i} (1-p)^{N -   i }. 
\end{equation}
In Algorithm~\ref{alg:calc}, we first calculate the summations of every subset, which can be viewed as a variant of the subset-sum problem~\cite{IA}. For each $i$ ($1 \leq i \leq N$), we merge the same-value summation results of the subsets $S_j^{(i)}$ and count the total number of $S_j^{(i)}$ which have that summation value. Then for the subsets that have summation results larger than $n-k$, we can calculate the corresponding failure probability according to equation~(\ref{eq:ps}). In the algorithm, $T,L,R$ represent three auxiliary lists for subset summation. For a auxiliary list $X$, we use $X.\mathsf{length}$ to denote the number of elements of the list, $X.\mathsf{index}$ to denote the current index number of  the list, $V_X(j)$ to denote the value of $j^{th}$ element in $X$, and $C_X(i,j)$ to denote the total number of subsets that have the same element number $i$ and the same summation value $V_X(j)$. Although the total number of subsets is $2^N$, Algorithm~\ref{alg:calc} is a polynomial time algorithm:

\begin{thm}
The complexity of Algorithm~\ref{alg:calc} is $\mathcal{O}(nN)$.
\end{thm}
\begin{proof}
Since the summation of a valid allocation $S$ itself is the largest in all the summations of the subsets of $S$, the element number $T.\mathsf{length}$ in $T$ cannot exceed $n$. Through the merge of subsets with the same summation values, each of the $N$ for-loops has the complexity $\mathcal{O}(n)$. So the total complexity is $\mathcal{O}(nN)$.
\end{proof}

\begin{algorithm}                                
\begin{spacing}{1.0}   
\begin{algorithmic}[1]               
    \Require a valid allocation $S(n,N,l)$, $(1 \leq l \leq \mathbb{P}(n,N))$
    \Ensure the failure probability $P_S$ of the allocation
    \Function{CalculateProbability}{$S(n,N,l)$}
    \State $\{n_1,n_2,\dots,n_N\} \Leftarrow$ sort the allocation $S(n,N,l)$ in \Statex \hskip15pt non-descending order
    \State $L \Leftarrow \{0\}$
    \State \Comment calculate summations of every subset   $\:\:\:\:\:\:\:\:\:\:\:\:\:\:\:\:\:\:\:\:\:\:\:\:\:$
    \For{$i=1 \to N$}    
        \State $T \Leftarrow  \phi $
        \State $R \Leftarrow L + n_i$, $C_R(1,1) \Leftarrow 1$
        \State $C_R(l,j) \Leftarrow  C_L(l-1,j)$, for all nonzero $C_L(l,j)$, \Statex \hskip30pt $2 \leq j \leq L.\mathsf{length}, 2 \leq l \leq i$, $i \ge 2$
        \State $L.\mathsf{index},R.\mathsf{index} \Leftarrow 1$  
        \While{$L.\mathsf{index} \leq L.\mathsf{length}$}
            \If {$V_L(L.\mathsf{index}) == V_R(R.\mathsf{index})$}
            \State $T \Leftarrow T \cup {V_R(R.\mathsf{index})}$
            \State {for all $1 \leq l \leq i$, $C_T(l,T.\mathsf{length}) \Leftarrow $\Statex \hskip58pt $C_L(l,L.\mathsf{index}) + C_R(l,R.\mathsf{index})$}
            \State increase $L.\mathsf{index}, R.\mathsf{index}$ by 1
            \Else
            \If  {$V_L(L.\mathsf{index}) < V_R(R.\mathsf{index})$}
            \State $T \Leftarrow T \cup {V_L(L.\mathsf{index})}$        
            \State $C_T(l,T.\mathsf{length}) \Leftarrow C_L(l,L.\mathsf{index})$, for \Statex \hskip72pt all $1 \leq l \leq i$    
            \State increase $L.\mathsf{index}$ by 1
            \Else
            \State $T \Leftarrow T \cup {V_R(R.\mathsf{index})}$
            \State $C_T(l,T.\mathsf{length}) \Leftarrow C_R(l,R.\mathsf{index})$, for \Statex \hskip72pt all $1 \leq l \leq i$
            \State increase $R.\mathsf{index}$ by 1
            \EndIf
            \EndIf
        \EndWhile
        \State $oldLn \Leftarrow T.\mathsf{length}$
        \State $T \Leftarrow T \cup \{V_R(R.\mathsf{index}), V_R(R.\mathsf{index} + 1),\dots, $ \Statex \hskip30pt $ V_R(R.\mathsf{length})\}$
        \State $\{C_T(l,oldLn+1), \dots,C_T(l,T.\mathsf{length})\}  \Leftarrow $ \Statex \hskip30pt $ \{C_R(l,R.\mathsf{index}),\dots,C_R(l,R.\mathsf{length})\}$, $1 \leq l \leq i$                   
        \State $L \Leftarrow T$
    \EndFor
    \State $P_S \Leftarrow 0$
    \State   \Comment count the number of subsets with the summation $\:\:$\Statex \hskip15pt results larger than $n-k$ 
    \For{$i=1 \to N$}      
        \State $sum \Leftarrow 0$
        \For {$j=1 \to T.\mathsf{length}$}
            \If {$V_T(j)>n-k$}
            \State $sum \Leftarrow sum + C_T(i,j)$
            \EndIf
        \EndFor
        \State $P_S \Leftarrow P_S + sum \times p^i(1-p)^{N-i}$
    \EndFor 
    \EndFunction
\end{algorithmic}
\end{spacing}
\caption{OSA scheme - stage II}          
\label{alg:calc}                  
\end{algorithm}

\begin{thm}
Algorithm~\ref{alg:calc} can output the failure probability $P_S$ for the input allocation.
\end{thm}

\begin{proof}
In line 3 we initialize the auxiliary list $L$ with an empty element $'0'$, representing the summation result of $0$ element of the input allocation $S$. Line 4 to line 31 calculate the summations of every subset of the input allocation. At the beginning of each round $i$ of the for loop $i=1,\dots,n$, the auxiliary list $L$ is the list containing the summation results of every subset of the first $l$ ($0 \leq l < i$) elements of the input allocation $S$. Line 7 to line 8 calculate the auxiliary list $R$ by adding the new element $n_i$ to $L$: $R = L + n_i$. Since the first element in $L$ is the empty $'0'$, $C_R(1,1)$ will be $1$, indicating that the total number of subsets that have $1$ element and summation value $n_i$ is $1$. Then the rest value of $C_R(l,j)$ will be $C_L(l-1,j)$ for $2 \leq j \leq L.\mathsf{length}$ because of the addition of $n_i$ to $L$. The elements of allocation $S$ are sorted in non-descending order, thus the elements in both $L$ and $R$ are also in non-descending order. From line 10 to line 29, we merge the elements of the auxiliary lists $L$ and $R$ into a temporary auxiliary list $T$ one by one, following the rules below:

\begin{itemize}
\item If the value of the current element $V_L(L.\mathsf{index})$ in $L$ is equal to the current element $V_R(R.\mathsf{index})$ in $R$, add the value into $T$. The corresponding counter $C_T(l,T.\mathsf{length})$ is equal to the sum of the two counters: $C_T(l,T.\mathsf{length}) = C_L(l,L.\mathsf{index}) + C_R(l,R.\mathsf{index})$ for $1 \leq l \leq i$.
\item If the value of the current element $V_L(L.\mathsf{index})$ in $L$ is smaller than the current element in $R$, add the element $V_L(L.\mathsf{index})$ into $T$. Set the counter $C_T(l,T.\mathsf{length})$ to $C_L(l, L.\mathsf{index})$ for $1 \leq l \leq i$.
\item If the value of the current element $V_R(R.\mathsf{index})$ in $R$ is smaller than the current element in $L$, add the element $V_R(R.\mathsf{index})$ into $T$. Set the counter $C_T(l,T.\mathsf{length})$ to $C_R(l, R.\mathsf{index})$ for $1 \leq l \leq i$.
\item Since the last element in $L$ is smaller than some elements in $R$, after merging $L$ into $T$, we can directly merge the remaining elements of $R$ into $T$ through line 28 to line 29.
\end{itemize}
At the end of each for loop, the merged list $T$ is assigned back to $L$ for the next round of calculation. After $N^{th}$ round, list $T$ has the summation results of all the subsets in $S$.

Then the failure probability of $S$ can be easily calculated from line 34 to line 42 by counting the number of subsets with the summation results larger than $n-k$.
\end{proof}

Fig.~\ref{fig:sumsubset} illustrates the summations for all the subsets of $S=\{1,2,2\}$. For $i=1$, $L = \{0\}$, $C_L(1,1)=0$, $R=\{1\}$, $C_R(1,1)=1$. The merged list $T=\{0,1\}$, $C_T =\{0,1\}$. For the second round, $L,C_L$ are assigned the values of $T,C_T$. According to line 7 and line 8 of Algorithm~\ref{alg:calc}, $R=L+n_2=\{2,3\}$ and $C_R(2,2) = C_L(1,2) = 1$. At the end of the third round, we can get the summation results $T = \{1,2,3,4,5\}$ and the counter matrix $C_T$, which correctly record the number of subsets that have the same summation value. As an example, $C_T(2,3) = 2$, indicating that there are two 2-element subsets ($\{n_1=1,n_2=2\}, \{n_1=1,n_3=2\}$) that have the same summation value $V_T(3) = 3$.

\begin{figure}
\centering
\includegraphics[width=0.8\columnwidth]{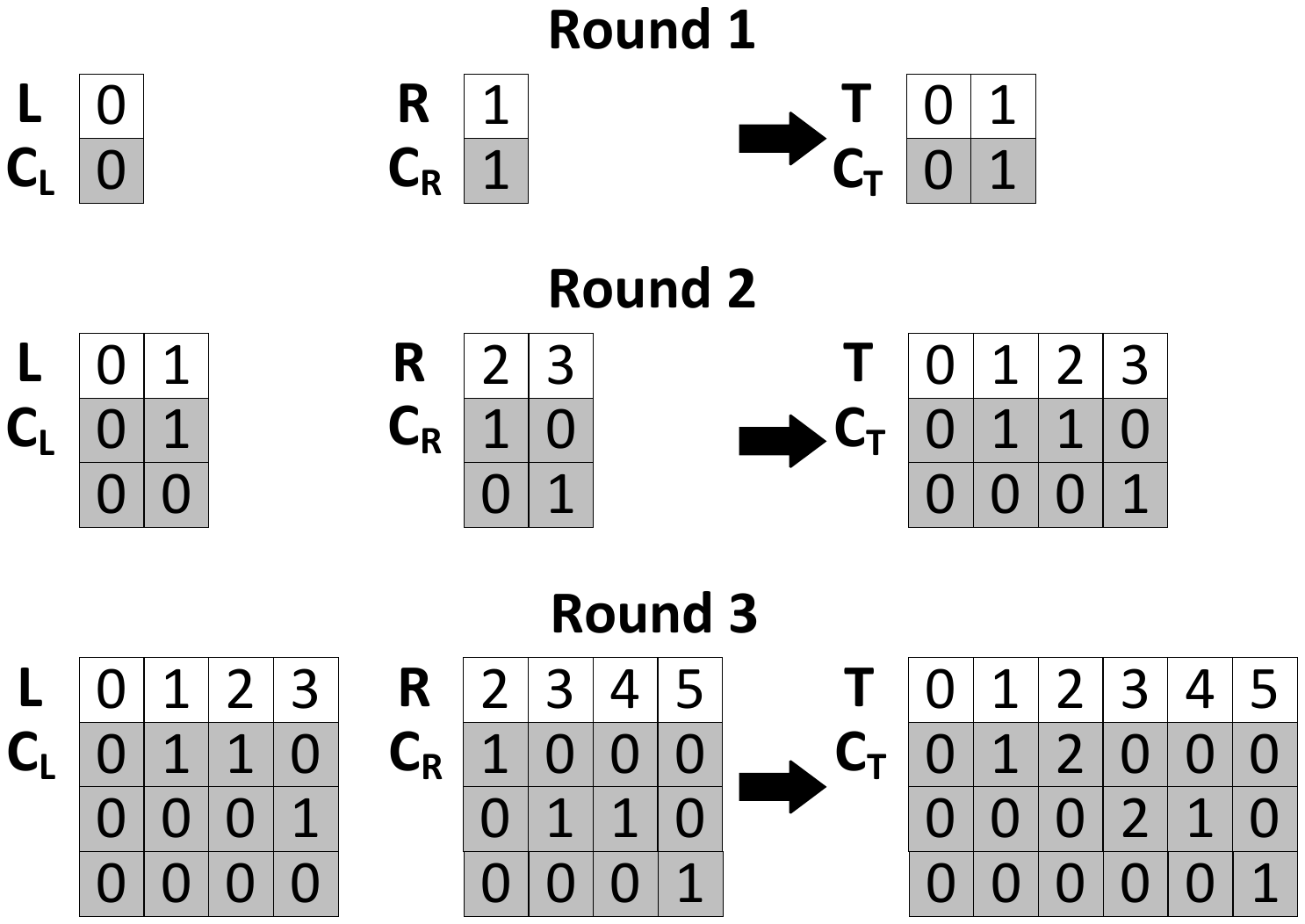}
\caption{The calculation of Algorithm~\ref{alg:calc} for $S=\{1,2,2\}$}
\label{fig:sumsubset}
\end{figure}

\subsubsection{OSA Scheme}
Based on the algorithms of the two stages, we can achieve the optimal storage allocation through Algorithm~\ref{alg:osa}. And it is straightforward to see:

\begin{thm}
The OSA scheme is a polynomial time algorithm.
\end{thm}

\begin{algorithm}                                
\begin{spacing}{1.0}   
\begin{algorithmic}               
    \Require the number of encoded parts $n$ and the number of storage centers $N$
    \Ensure the allocation with the lowest failure probability
    \Function{OSA}{$n,N$}
    \State  $S(n,N,l) \Leftarrow$ \Call{FindAllAllocations}{$n,N$} $(1 \leq l \leq$ \\ \hskip15pt $\mathbb{P}(n,N))$    
    \For{$l=1 \to \mathbb{P}(n,N)$}
        \State $P_S \Leftarrow$ \Call{CalculateProbability}{$S(n,N,l)$}
    \EndFor
    \State output the allocation with the lowest $P_S$
    \EndFunction
\end{algorithmic}
\end{spacing}
\caption{OSA scheme}          
\label{alg:osa}                  
\end{algorithm}

\subsection{Simulation Results for the OSA Scheme}

In this section we will show the performance of the OSA scheme for the given regenerating code with parameters $(n,k,d,\alpha,\beta,B)$ and number of data centers $N$. 

In Fig.~\ref{fig:simk} are the simulation results for $n=45$, $k=\{16,21,26,31\}$, $N=9$ and $p=0.1$. For performance comparison, we also plot the results for the even allocation as defined in equation~(\ref{eq:even}), where $\lfloor n/N \rfloor$ is the floor operation to get the largest integer that is less or equal to $n/N$, $mod(n,N)$ is the modulo operation to find the remainder of the division of $n$ by $N$. 
\begin{equation}
\label{eq:even}
n_i=\left\{
\begin{array}{cl}
\lfloor n/N \rfloor + 1, & 1 \leq i \leq mod(n,N)  \\ 
\lfloor n/N \rfloor, & mod(n,N) < i < N
\end{array}\right.
\end{equation}
The even allocation is a natural allocation scheme to store equal number of data blocks into each storage center. From the figure we can see that the failure probability of the OSA scheme is about half order of magnitude lower than the even allocation. Both of the probabilities will become higher when $k$ increases since there is less redundancy in the distributed cloud storage to recover the failed storage centers.

\begin{figure}
\centering
\includegraphics[width=1.0\columnwidth]{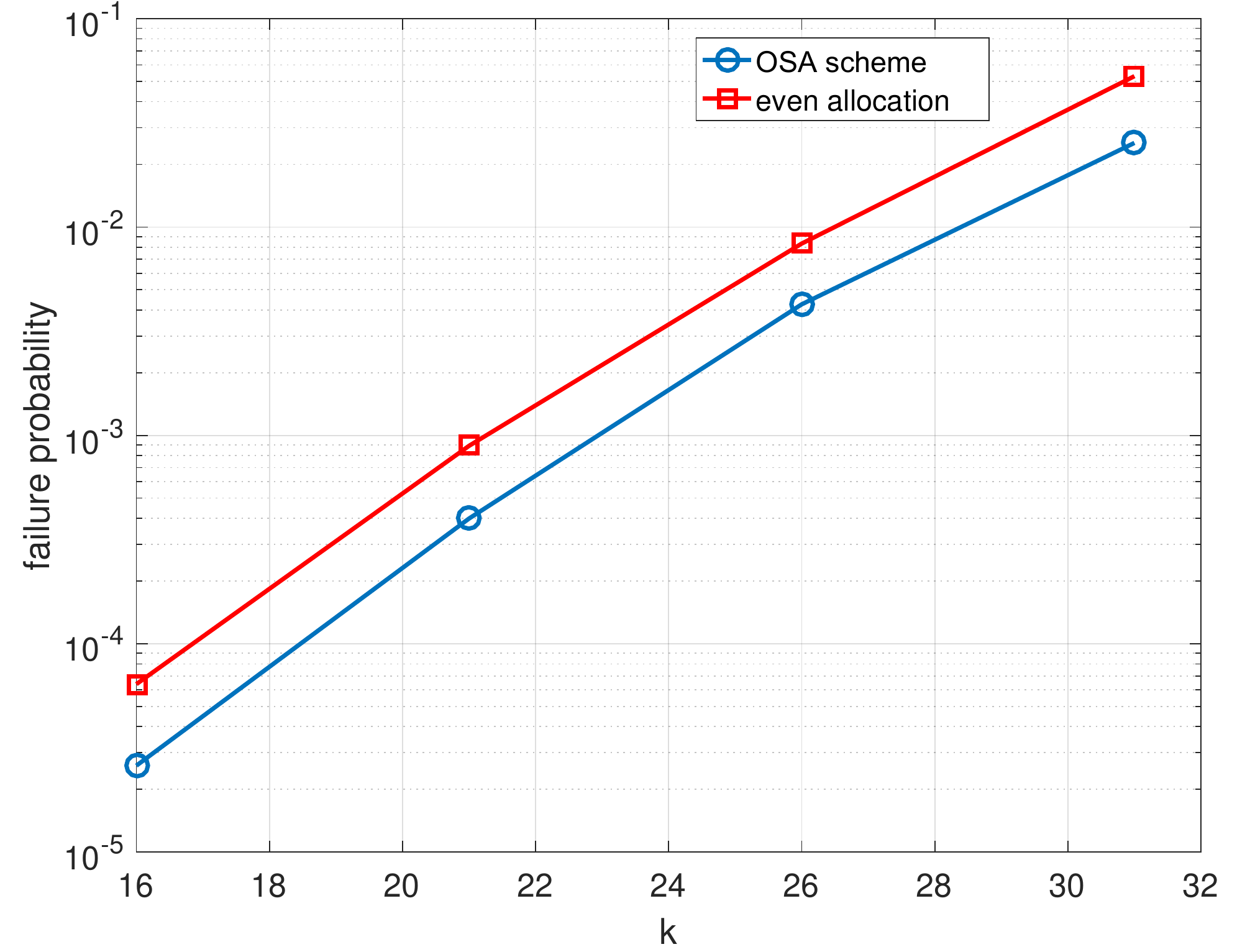}
\caption{Performance of the optimal storage allocation for different k}
\label{fig:simk}
\end{figure}

In Fig.~\ref{fig:simN} are the simulation results for $n=45$, $k=21$, $p=0.1$. In this simulation, we change the number of storage centers $N$ to study its impact to the failure probability. From the figure we can see that the failure probability will become lower when the number of storage centers increases. And the performance gap of the even allocation and the OSA scheme will diminish with the increasing of the number of storage centers.

\begin{figure}
\centering
\includegraphics[width=1.0\columnwidth]{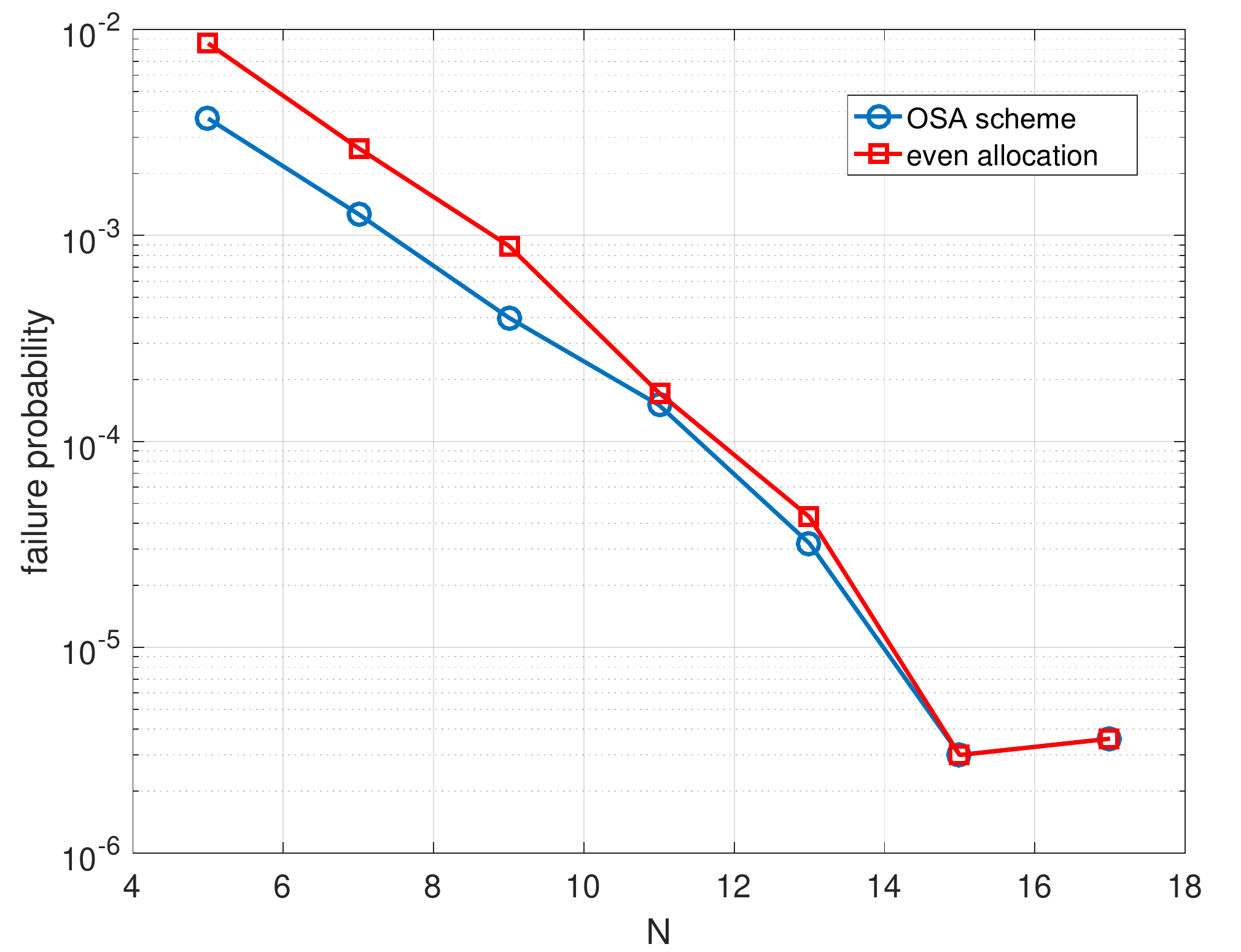}
\caption{Performance of the optimal storage allocation for different number of storage centers}
\label{fig:simN}
\end{figure}

\section{conclusion}\label{sec:conclusion}
In this paper we 
analyze the important applications of network coding in the IoT core network and the distributed cloud storage that stores the data generated by the IoT core network. We propose an adaptive network coding (ANC) scheme in the IoT core network with software defined wireless network (SDWN). Simulation results have demonstrated that the ANC scheme can achieve higher transmission efficiency than existing schemes. Then we introduce the optimal storage allocation problem for the distributed cloud storage that utilizes network coding. we propose an optimal storage allocation (OSA) scheme to solve the problem in polynomial time. We also conduct simulations to show that the OSA scheme can greatly improve the storage reliability. Impressed by the simplicity and efficacy of network coding in both communication and storage, we believe that more and more potential applications of network coding would be found and studied during the development of 
Internet of things to accelerate the whole deployment process.


%

%

%
%

\bibliography{bibdata}
\bibliographystyle{ieeetr}

%

%




\end{document}